\newcommand{\CLASSINPUTtoptextmargin}{0.73in}
\newcommand{\CLASSINPUTbottomtextmargin}{0.92in}
\documentclass[conference]{IEEEtran}
\IEEEoverridecommandlockouts

\usepackage[utf8]{inputenc}
\usepackage{amsfonts,amsmath,amssymb,amsthm,mathrsfs,bm,bbm}
\usepackage{enumerate}
\usepackage{cite}
\usepackage{graphicx}
\usepackage{stfloats}
\usepackage[caption=false,font=footnotesize]{subfig}
\usepackage{algpseudocode}
\usepackage{overpic}

\newtheorem{proposition}{Proposition}

\newtheorem{remark}{Remark}

\usepackage{hyperref}
\hypersetup{
  colorlinks   = true,    
  urlcolor     = blue,    
  linkcolor    = black,    
  citecolor    = red      
}

\newcommand{\eqdef}{:=}

\newcommand{\rvec}[1]{\mathbbm{#1}} 		
\newcommand{\rmat}[1]{\mathbbm{#1}} 	
\newcommand{\E}{\mathsf{E}}		
\newcommand{\V}{\mathsf{V}}			
\newcommand{\stdset}[1]{\mathbbmss{#1}}	
\newcommand{\set}[1]{\mathcal{#1}}		
\renewcommand{\vec}[1]{\bm{#1}}		
\newcommand{\CN}{\mathcal{CN}}			
\newcommand{\herm}{\mathsf{H}}			

\author{\vspace{-0.2cm}\IEEEauthorblockN{Lorenzo Miretti and S\l{}awomir Sta\'nczak}\vspace{0.2cm}
\IEEEauthorblockA{\emph{Technische Universität Berlin} and \emph{Fraunhofer Institute for Telecommunications Heinrich-Hertz-Institut}, Berlin, Germany \\
miretti@tu-berlin.de, slawomir.stanczak@hhi.fraunhofer.de}}

\title{Unlocking the Potential of Local CSI in \\ Cell-Free Networks with Channel Aging and Fronthaul Delays 
\thanks{L.~Miretti and S.~Sta\'nczak acknowledge the financial support by the Federal Ministry of Education and Research of Germany in the programme of “Souverän. Digital. Vernetzt.” Joint project 6G-RIC, project identification number: 16KISK020K and 16KISK030.} \thanks{© 2024 IEEE.  Personal use of this material is permitted.  Permission from IEEE must be obtained for all other uses, in any current or future media, including reprinting/republishing this material for advertising or promotional purposes, creating new collective works, for resale or redistribution to servers or lists, or reuse of any copyrighted component of this work in other works.}
}

\begin{document}
\def\baselinestretch{.99}
\setlength{\belowdisplayskip}{1pt}
\setlength{\belowdisplayshortskip}{1pt}
\setlength{\abovedisplayskip}{1pt}
\setlength{\abovedisplayshortskip}{1pt}

\maketitle

\begin{abstract}
It is generally believed that downlink cell-free networks perform best under centralized implementations where the local channel state information (CSI) acquired by the access-points (AP) is forwarded to one or more central processing units (CPU) for the computation of the joint precoders based on global CSI. However, mostly due to limited fronthaul capabilities, this procedure incurs some delay that may lead to partially outdated precoding decisions and hence performance degradation. In some scenarios, this may even lead to worse performance than distributed implementations where the precoders are locally computed by the APs based on partial yet timely local CSI. To address this issue, this study considers the problem of robust precoding design merging the benefits of timely local CSI and delayed global CSI. As main result, we provide a novel distributed precoding design based on the recently proposed \textit{team} minimum mean-square error method. As a byproduct, we also obtain novel insights related to the AP-CPU functional split problem. Our main conclusion, corroborated by simulations, is that the opportunity of performing some local precoding computations at the APs should not be neglected, even in centralized implementations. 
\end{abstract}
\vspace{-0.4cm}

\section{Introduction}
Cell-free massive MIMO is one of the most promising candidate technologies for enhancing the performance of future generation wireless networks \cite{demir2021foundations}. Most of the related current research effort focuses on the development of practical methods and architectures for turning the known theoretical gains of coordinated multi-point concepts into commercially attractive solutions \cite{ngo2017cellfree,bashar2019quantization,interdonato2019ubiquitous,hu2019adc,lozano2020fractional,buzzi2020,gottsch2022subspace}. Of particular relevance is the debate on the type of joint precoding and combining implementation, and on its impact on the functional split problem in cloud radio access network (C-RAN) architectures \cite{kang2015fronthaul}. 

More specifically, taken aside promising yet exotic schemes based, e.g., on sequential processing over serial fronthauls \cite{shaik2021mmse,miretti2021team} or iterative bidirectional over-the-air processing \cite{italo2021ota}, this debate is essentially centered around the comparison between fully distributed and centralized implementations. In fully distributed implementations, the access points (AP) locally compute their precoders and combiners based on local channel state information (CSI) only \cite{ngo2017cellfree}. In contrast, in centralized implementations, these functions are moved to one or more central processing units (CPU) endowed with global CSI \cite{demir2021foundations}. Hence, due to their ability to form joint precoders and combiners based on a broader view of the channel state, centralized implementations are often considered superior, especially in terms of spectral efficiency. However, the theoretical superiority of centralized implementations is typically shown under ideal assumptions such as those related to fronthaul capabilities, which can be challenged in many practical deployments.

\vspace{-0.1cm}Against this background, this paper studies centralized downlink cell-free networks with fronthaul delays, which is a key impairment in real-world settings. In particular, we focus on scenarios where, due to fronthaul limitations and mobility, the delay incurred by the CPUs in collecting global CSI, computing the precoders, and forwarding the result is non-negligible with respect to channel aging \cite{heath2013aging,zheng2021aging,schotten2022delay}. Intuitively, in these scenarios, centralized implementations experience performance degradation and may be even outperformed by fully distributed implementations with precoders formed using partial yet more timely local CSI. To adress this issue, we formulate a distributed precoding design problem that aims to jointly exploit timely local CSI and delayed global CSI. To the best of our knowledge, this is the first time that a similar problem is addressed in the literature. Then, we derive an optimal solution based on a novel application of the recent \textit{team} theoretical framework \cite{miretti2021team}. In addition, we discuss the structure of the optimal solution and related practical implementation aspects in C-RAN architectures. Interestingly, our theoretical and numerical results show that carefully designed implementations that delegate the computation of at least a portion of the precoders to the APs may significantly outperform both  centralized and fully distributed implementations, even under pedestrian mobility and relatively small delays.   


\textit{Notation:}
We denote by $\stdset{R}_{++}$ the set of positive reals. The Euclidean and Frobenius norms are denoted by $\|\cdot\|$ and $\|\cdot\|_{\mathsf{F}}$, respectively. Let $(\Omega,\Sigma,\mathbb{P})$ be a probability space. We denote by $\set{H}^K$ the set of random vectors, i.e., $K$-tuples of $\Sigma$-measurable functions $\Omega \to \stdset{C}$, satisfying $(\forall \rvec{x}\in \set{H}^K)$ $\E[\|\rvec{x}\|^2]<\infty$. Given a random variable $X\in \set{H}$, we denote by $\E[X]$ and $\V(X)$ its expected value and variance, respectively. 

\section{Problem statement}
\label{sec:statement}
\subsection{Basic definitions and assumptions}
We consider the downlink of a cell-free wireless network \cite{demir2021foundations} composed of $L$ APs indexed by $\mathcal{L}:=\{1,\ldots,L\}$, each of them equipped with $N$ antennas, and $K$ single-antenna user equipments (UEs) indexed by $\mathcal{K}:=\{1,\ldots,K\}$. By assuming a standard flat-fading channel model for each time-frequency resource element, we denote an arbitrary realization of the $(NL\times K)$-dimensional global channel matrix by
\begin{equation*}
\rvec{H} \eqdef \begin{bmatrix}
\rvec{h}_1 & \ldots & \rvec{h}_K
\end{bmatrix} = \begin{bmatrix}
\rvec{H}_1 \\ \vdots \\ \rvec{H}_L
\end{bmatrix} = \begin{bmatrix}
\rvec{h}_{1,1} & \ldots & \rvec{h}_{1,K} \\
\vdots &  \ddots & \vdots \\
\rvec{h}_{L,1} & \ldots & \rvec{h}_{L,K}
\end{bmatrix},
\end{equation*}
where $\rvec{h}_{l,k} \in\set{H}^{N}$ is a random vector modeling the fading state between AP~$l\in \set{L}$ and UE~$k\in \set{K}$. As customary in the channel aging literature \cite{heath2013aging, zheng2021aging}, we assume that the random channel realizations $\rvec{H}$ evolve over time according to a stationary and ergodic discrete-time random process $\{\rvec{H}[t]\}_{t\in \stdset{Z}}$, without any further specific assumption on the time correlation. In addition, we assume that the portions $\{\rvec{h}_{l,k}[t]\}_{t\in \stdset{Z}}$ of $\{\rvec{H}[t]\}_{t\in \stdset{Z}}$ corresponding to different AP-UE pairs are mutually independent random processes. 

Similarly, by focusing on simple multi-user cooperative transmission techniques based on linear precoding and on treating interference as noise \cite{demir2021foundations}, we denote an arbitrary realization of the $(NL\times K)$-dimensional joint precoding matrix by
\begin{equation*}
\rvec{T} = \begin{bmatrix}
\rvec{t}_1 & \ldots & \rvec{t}_K
\end{bmatrix} = \begin{bmatrix}
\rvec{T}_1 \\ \vdots \\ \rvec{T}_L
\end{bmatrix} =  \begin{bmatrix}
\rvec{t}_{1,1} & \ldots & \rvec{t}_{1,K} \\
\vdots &  \ddots & \vdots \\
\rvec{t}_{L,1} & \ldots & \rvec{t}_{L,K}
\end{bmatrix},
\end{equation*}
where $\rvec{t}_{l,k} \in\set{H}^{N}$ is a linear precoding vector applied by AP~$l\in \set{L}$ to the coded and modulated data stream for UE~$k\in \set{K}$. The joint precoding matrix $\rvec{T}$ evolves over time according to a random process $\{\rvec{T}[t]\}_{t\in \stdset{Z}}$, where $\rvec{T}[t]$ is adapted to the random channel realization $\rvec{H}[t]$ based on the available channel state information (CSI) at time $t\in \stdset{Z}$. In particular, we focus on the case where the submatrices $\rvec{T}_l[t]$ of $\rvec{T}[t]$ corresponding to the precoding matrices applied by each AP $l\in \set{L}$ may depend on different CSI. This aspect is treated in more details next. 

\subsection{Delayed CSI sharing}
Canonical cell-free network models based on time-division duplex operations assume each AP $l\in \set{L}$ to acquire local measurements of the downlink local channel $\rvec{H}_l = [\rvec{h}_{l,1} \; \ldots \; \rvec{h}_{l,K}]$ 
by means of uplink pilot signals. In centralized implementations \cite{demir2021foundations}, these local measurements are then typically forwarded by the APs to one or more central processors for the computation of the joint precoding matrix based on measurements of the global channel $\rvec{H}$. However, this process inevitably incurs some delay and may lead to outdated precoding decisions in many practical scenarios, making distributed implementations based on timely local measurements \cite{ngo2017cellfree} a competitive alternative. 
This observation is also corroborated by our simulations in Section~\ref{sec:sim} for relatively small delays.

In this work, we study the impact of delayed CSI sharing on performance and robust precoding design by considering the following simplified model. We assume that each AP~$l\in \set{L}$ can form its precoding matrix based on perfect instantaneous knowledge of the local channel $\rvec{H}_l[t]$, and perfect $d$-step delayed knowledge of the global channel $\rvec{H}[t]$. More precisely, we assume that the precoders $\rvec{T}_l[t]$ of AP $l\in \set{L}$ at time $t\in \stdset{Z}$ are constrained to be functions of the CSI
\begin{equation}
\label{eq:CSI}
S_l[t] \eqdef (\rvec{H}_l[t],Z[t]), \quad Z[t] \eqdef \rvec{H}[t-d], \quad d\in \stdset{N}.
\end{equation}
As we will see, the key feature of the above model is that it allows us to design robust precoders that combine the benefits of (delayed) centralized interference management with the opportunity of performing timely local refinements.     
\begin{remark}
To avoid technical digressions, cumbersome notation, and to better focus on the essence of the problem, in this study we do not consider aspects such as channel estimation errors, user-centric network clustering, and the opportunity of storing and exploiting CSI history such as $\rvec{H}[t-i]$ for $i> d$. However, these aspects can be easily incorporated in our model and results following the approach in \cite{miretti2023duality}, and will be covered in details in an extended version of this study. 
\end{remark}

\begin{remark}\label{rem:computation}
Our model and main derivations do not explicitly consider centralized precoding computation, and assume a distributed system where all precoders are locally computed by the APs after a preliminary CSI sharing step. However, we remark that all computations involving $\rvec{H}[t-d]$ only can also be implemented on a central processor. Hence, our model implicitly covers both the aforementioned centralized and fully distributed implementations, which correspond to extreme functions that discard either $\rvec{H}_l[t]$ or $\rvec{H}[t-d]$, respectively. More interestingly, our model also covers intermediate cases where the computation of the precoders is split among a central processor operating on the basis of global delayed CSI, and the APs operating on the basis of timely local CSI. Additional details are given in Section~\ref{ssec:implementation}.
\end{remark}

\subsection{Team MMSE precoding}
Following the approach in \cite{miretti2021team}, which introduced a novel non-heuristic method for optimal distributed  precoding design when the APs are endowed with different CSI, we consider the following parametric \textit{Team MMSE} precoding problem:
\begin{equation}
\label{prob:TMMSE}
\underset{\rvec{T} \in \set{T}}{\text{minimize}}~\E\left[\|\vec{P}^{\frac{1}{2}}\rvec{H}^\herm\rvec{T}-\vec{I}_K\|_{\mathsf{F}}^2\right] + \sum_{l=1}^L\sigma_l\E\left[\|\rvec{T}_l\|_{\mathsf{F}}^2\right],
\end{equation}
where $\set{T}\subseteq \set{H}^{K\times LN}$ is a given \textit{information} constraint \cite{miretti2021team, miretti2023duality} induced by the CSI structure \eqref{eq:CSI}, $\vec{P}=\text{diag}(\vec{p})$ with $\vec{p} = (p_1,\ldots,p_K) \in \stdset{R}_{++}^K$, and $\vec{\sigma}=(\sigma_1,\ldots,\sigma_L)\in \stdset{R}_{++}^L$ are given  parameters.\footnote{Note that the Team MMSE precoding problem is equivalently formulated in \cite{miretti2021team,miretti2023duality} as $K$ separate problems for each of the $K$ precoding vectors $(\rvec{t}_k)_{k=1}^K$, coupled by the problem parameters $(\vec{p},\vec{\sigma})$. The difference between \cite{miretti2021team} and \cite{miretti2023duality} is that \cite{miretti2021team} focuses on the case $\vec{\sigma}=\vec{1}$.} Note that, to improve readability of the paper, we  omit the dependency on the time index~$t$ since $\{\rvec{H}[t],S_1[t],\ldots,S_L[t]\}_{t\in \stdset{Z}}$ is a stationary random process. However, we remark that the impact of the delay $d$ is still fully captured in \eqref{eq:TMMSE} by means of the  constraint set $\set{T}$. 

Informally, the role of the constraint set $\set{T}$ is to enforce the precoders of each AP $l\in \set{L}$ to be functions of $S_l = (\rvec{H}_l,Z)$ only, where $S_l$ denotes a realization of \eqref{eq:CSI} at some arbitrary time $t\in \stdset{Z}$. As proposed in \cite{miretti2021team, miretti2023duality}, we formally define the constraint set $\set{T}$ as follows: we
let \begin{equation*}
(\forall k \in \set{K})~\rvec{t}_k\in \set{T}_k \eqdef \set{H}_1^N \times \ldots \times \set{H}_L^N,
\end{equation*}
where $\set{H}_l^N\subseteq \set{H}^N$ denotes the set of $N$-tuples of $\Sigma_l$-measurable functions $\Omega \to \stdset{C}$ satisfying $(\forall \rvec{x}\in \set{H}_l^N)$ $\E[\|\rvec{x}\|^2]<\infty$, and where $\Sigma_l \subseteq \Sigma$ is the sub-$\sigma$-algebra induced by the CSI $S_l = (\rvec{H}_l,Z)$ available at AP $l\in \set{L}$. In the team theoretical literature, $\Sigma_l$ is also called the \emph{information subfield} of AP $l$. Then, we let $\set{T}\eqdef \set{T}_1\times \ldots \times \set{T}_K$. The interested reader is referred to \cite{yukselbook} for an introduction to the measure theoretical notions used in the above definitions. However, we stress that these notions are by no means required for understanding the key results of this study.

\begin{remark}
Problem~\eqref{prob:TMMSE} can be motivated under multiple points of view. For instance, the solution to Problem~\eqref{prob:TMMSE} can be interpreted as the best distributed approximation of regularized channel inversion (recovered for $d = 0$), where the parameters $(\vec{p},\vec{\sigma})$ can be tuned to balance UE priorities and APs power consumption. Furthermore, solving \eqref{prob:TMMSE} corresponds to minimizing the individual MSE between the transmit and receive data-bearing symbols after linear processing over a dual uplink channel with UE uplink powers $\vec{p}$ and AP noise powers~$\vec{\sigma}$. Finally, an information theoretical motivation is obtained by evaluating performance using the so-called \emph{hardening} inner bound \cite{demir2021foundations} on the ergodic capacity region. Specifically, by leveraging the known uplink-downlink duality principle for fading channels under a sum power constraint (see, e.g., \cite{demir2021foundations}), \cite{miretti2021team} proves that, by choosing $\vec{\sigma} = \vec{1}$ and $\vec{p}$ such that $\sum_{k=1}^Kp_l = P$, the solution to \eqref{prob:TMMSE} is Pareto optimal, in the sense that it produces rate tuples on the boundary of the considered inner bound under a sum power constraint $P$ (and unitary noise powers). Conversely, \cite{miretti2021team} shows that all boundary points under a sum power constraint $P$ can be achieved by solutions to \eqref{prob:TMMSE} for $\vec{\sigma} = \vec{1}$ and for some $\vec{p}$ such that $\sum_{k=1}^Kp_l = P$. Furthermore, \cite{miretti2023duality} proves that all boundary points under per-AP power constraints can be achieved by solutions to \eqref{prob:TMMSE} for some $(\vec{p},\vec{\sigma})$. 
\end{remark}

The parameters $(\vec{p},\vec{\sigma})$ of Problem~\eqref{prob:TMMSE} can be tuned to maximize some network utility function under some power and/or quality of service constraints as in, e.g, \cite{miretti2022joint, miretti2023duality}, or set heuristically as for the many variants of the MMSE or regularized zero forcing precoding schemes~\cite{demir2021foundations}. Additional details on the tuning of these parameters are left for the extended version of this study, since they mostly relate to resource allocation and power control aspects that are not specific to the delayed CSI sharing model~\eqref{eq:CSI}. The rest of this study is devoted to solving the Team MMSE precoding problem \eqref{prob:TMMSE} under the considered delayed CSI sharing model~\eqref{eq:CSI}, for arbitrary problem parameters $(\vec{p},\vec{\sigma})$. Only our simulations in  Section~\ref{sec:sim} will focus on a specific example of $(\vec{p},\vec{\sigma})$.

\section{Problem solution}
\subsection{Optimal solution}
Problem~\eqref{prob:TMMSE} is a functional (i.e., infinite dimensional) optimization problem belonging to the class of \textit{team decision} problems \cite{yukselbook}, which are notoriously difficult, even when convex as in our case. The main difficulty lies in the information constraint $\set{T}$, which prevents the direct application of standard methods and numerical routines for finite dimensional convex problems. However, \cite{miretti2021team} showed that Problem~\eqref{prob:TMMSE} can be mapped to a minor variation of the subclass of \textit{quadratic} team problems \cite{yukselbook}, and, as a consequence, that the following necessary and sufficient optimality conditions hold.
\begin{proposition}
\label{prop:TMMSE}
For given $\vec{p}\in\stdset{R}_{++}^K$ and $\vec{\sigma}\in \stdset{R}_{++}^L$, Problem~\eqref{prob:TMMSE} admits a unique solution, which is also the unique $\rvec{T} \in \set{T}$ satisfying
\begin{equation}\label{eq:TMMSE}
(\forall l \in \mathcal{L})~\rvec{T}_{l} =  \rmat{F}_l\left(\vec{I}_K-\sum_{j \in \set{L}\backslash  \{l\}} \vec{P}^{\frac{1}{2}}\E\left[\rmat{H}_{j}^\herm\rvec{T}_j\Big|S_l\right] \right), 
\end{equation} 
where $\rmat{F}_l:=\left(\rmat{H}_{l}\vec{P}\rmat{H}_{l}^\herm + \sigma_l\vec{I}_N\right)^{-1}\rmat{H}_{l}\vec{P}^{\frac{1}{2}} \in \set{H}^{N\times K}$. 
\end{proposition}
\begin{proof}
The proof follows readily by replacing $\vec{\sigma}=\vec{1}$ with an arbitrary $\vec{\sigma}\in \stdset{R}_{++}^L$ from the proof of \cite[Lemma~2]{miretti2021team}. 
Informally, \eqref{eq:TMMSE} is obtained by minimizing the objective in \eqref{prob:TMMSE} with respect to  $\rvec{T}_{l}$, and by fixing $\rvec{T}_j$ for $j\neq l$. This gives a set of necessary optimality conditions, related to the game theoretical notion of Nash equilibrium. The key step of the proof shows that these conditions are also sufficient.
\end{proof}

Proposition~\ref{prop:TMMSE} and its extensions covering channel estimation errors and user-centric network clustering are used in \cite{miretti2021team,miretti2023duality} to derive optimal distributed precoders under local CSI models of the type $(\forall l \in \mathcal{L})~S_l = \rvec{H}_l$ for fully distributed implementations, or under sequential CSI sharing models of the type $(\forall l \in \mathcal{L})~S_l = (\rvec{H}_1,\ldots,\rvec{H}_l)$ for partially distributed implementations exploiting the properties of serial fronthauls.
In the following, we use Proposition~\ref{prop:TMMSE} to derive the main result of this study, i.e., the solution to Problem~\eqref{prob:TMMSE} under the delayed CSI sharing model $(\forall l \in \mathcal{L})~S_l = (\rvec{H}_l,Z)$ in \eqref{eq:CSI}.

\begin{proposition}\label{prop:DTMMSE}
For given $\vec{p}\in\stdset{R}_{++}^K$ and $\vec{\sigma}\in \stdset{R}_{++}^L$, the unique solution to Problem~\eqref{prob:TMMSE} is given by
\begin{equation}\label{eq:DTMMSE}
(\forall l \in \mathcal{L})~\rvec{T}_{l} =  \rmat{F}_l\rvec{C}_l, 
\end{equation} 
where $\rmat{F}_l:=\left(\rmat{H}_{l}\vec{P}\rmat{H}_{l}^\herm + \sigma_l\vec{I}_N\right)^{-1}\rmat{H}_{l}\vec{P}^{\frac{1}{2}} \in \set{H}^{N\times K}$, and $\rvec{C}_l \in \set{H}^{K\times K}$ is given by the unique solution to the linear system of equations
\begin{equation}\label{eq:lin_syst}
(\forall l \in \mathcal{L})~\rvec{C}_l + \sum_{j\in \set{L}\backslash \{l\}}\E[\vec{P}^{\frac{1}{2}}\rvec{H}_j^\herm\rvec{F}_j|Z]\rvec{C}_j = \vec{I}_K.
\end{equation}
\end{proposition}
\begin{proof}(Sketch)
The proof follows by verifying that \eqref{eq:DTMMSE} satisfies \eqref{eq:TMMSE} via simple algebraic manipulations. 
\end{proof}

\subsection{Interpretation and computation of the optimal solution} 
\label{ssec:computation}
Proposition~\ref{prop:DTMMSE} states the optimality of the two-stage precoding structure in \eqref{eq:DTMMSE}, where each stage depends on a distinct portion of the CSI $S_l = (\rvec{H}_l,Z)$. Specifically, the optimal precoder $\rvec{T}_l$ at AP~$l\in \set{L}$ is composed by a first $N\times K$ \textit{local} MMSE precoding stage $\rvec{F}_l$ \cite{demir2021foundations}, function of the timely local CSI $\rvec{H}_l$, and a second $K\times K$ precoding stage $\rvec{C}_l$, function of the delayed global CSI $Z$. To better understand the dependency of $\rvec{C}_l$ on $Z$, we observe that the linear system of equations \eqref{eq:lin_syst} defining $\rvec{C}_l$ has random coefficients $\E[\vec{P}^{\frac{1}{2}}\rvec{H}_l^\herm\rvec{F}_l|Z]$ which are functions of $Z$. In particular, each realization of the $L$ precoding stages $(\rvec{C}_l)_{l=1}^L$ can be obtained by solving \eqref{eq:lin_syst} disjointly for each realization of $Z$. More precisely, a realization $(\vec{C}_l)_{l=1}^L$ of $(\rvec{C}_l)_{l=1}^L$ for a given realization $z$ of $Z$ is given by the solution to the finite dimensional linear system of equations 
\begin{equation*}
(\forall l \in \mathcal{L})~\vec{C}_l + \sum_{j\in \set{L}\backslash \{l\}}\E[\vec{P}^{\frac{1}{2}}\rvec{H}_j^\herm\rvec{F}_j|Z=z]\vec{C}_j = \vec{I}_K,
\end{equation*}
which can be computed using standard techniques, provided that the coefficients $\E[\vec{P}^{\frac{1}{2}}\rvec{H}_l^\herm\rvec{F}_l|Z=z]$ are known.

By reintroducing the time index~$t$, we notice that these coefficients take the form $\E[\vec{P}^{\frac{1}{2}}\rvec{H}_l[t]^\herm\rvec{F}_l[t]|Z[t]]=$
\begin{align*}
\E[\vec{P}^{\frac{1}{2}}\rvec{H}_l[t]^\herm\left(\rmat{H}_{l}[t]\vec{P}\rmat{H}_{l}[t]^\herm + \sigma_l\vec{I}_N\right)^{-1}\rmat{H}_{l}[t]\vec{P}^{\frac{1}{2}}|\rvec{H}_l[t-d]],
\end{align*}
i.e., they are functions of $\rvec{H}_l[t-d]$ defined by the conditional distribution of $\rvec{H}_l[t]$ given $\rvec{H}_l[t-d]$, which we recall is independent of $t$ due to stationarity. Importantly, these coefficients can be computed in parallel for each AP. 
Furthermore, we notice that they can be interpreted as the optimal $d$-step MMSE predictors of the $K\times K$ effective channels $\vec{P}^{\frac{1}{2}}\rvec{H}_l[t]^\herm\rvec{F}_l[t]$ after local MMSE precoding. Unfortunately, closed-form expressions for the coefficients $\E[\vec{P}^{\frac{1}{2}}\rvec{H}_l^\herm\rvec{F}_l|Z]$ may not be available in many practical cases. However, we remark that many approximate numerical methods taken from the vast literature on estimation/prediction theory could be potentially applied. Of particular practical interest are data driven techniques that do not require explicit knowledge of the conditional distribution of $\rvec{H}_l[t]$ given $\rvec{H}_l[t-d]$. For simplicity, in this work, we evaluate numerically each expectation using an empirical average over a sample set generated according to the conditional distribution of $\rvec{H}_l[t]$ given $\rvec{H}_l[t-d]$, assumed known. In addition, in Section~\ref{ssec:approximate}, we discuss some suboptimal approximations. We leave the evaluation of more advanced techniques as a promising future line of research. 

\subsection{C-RAN functional split aspects}
\label{ssec:implementation}
\begin{figure*}[!t]
\centerline{\subfloat[]{\begin{overpic}[width=2in,tics=20]{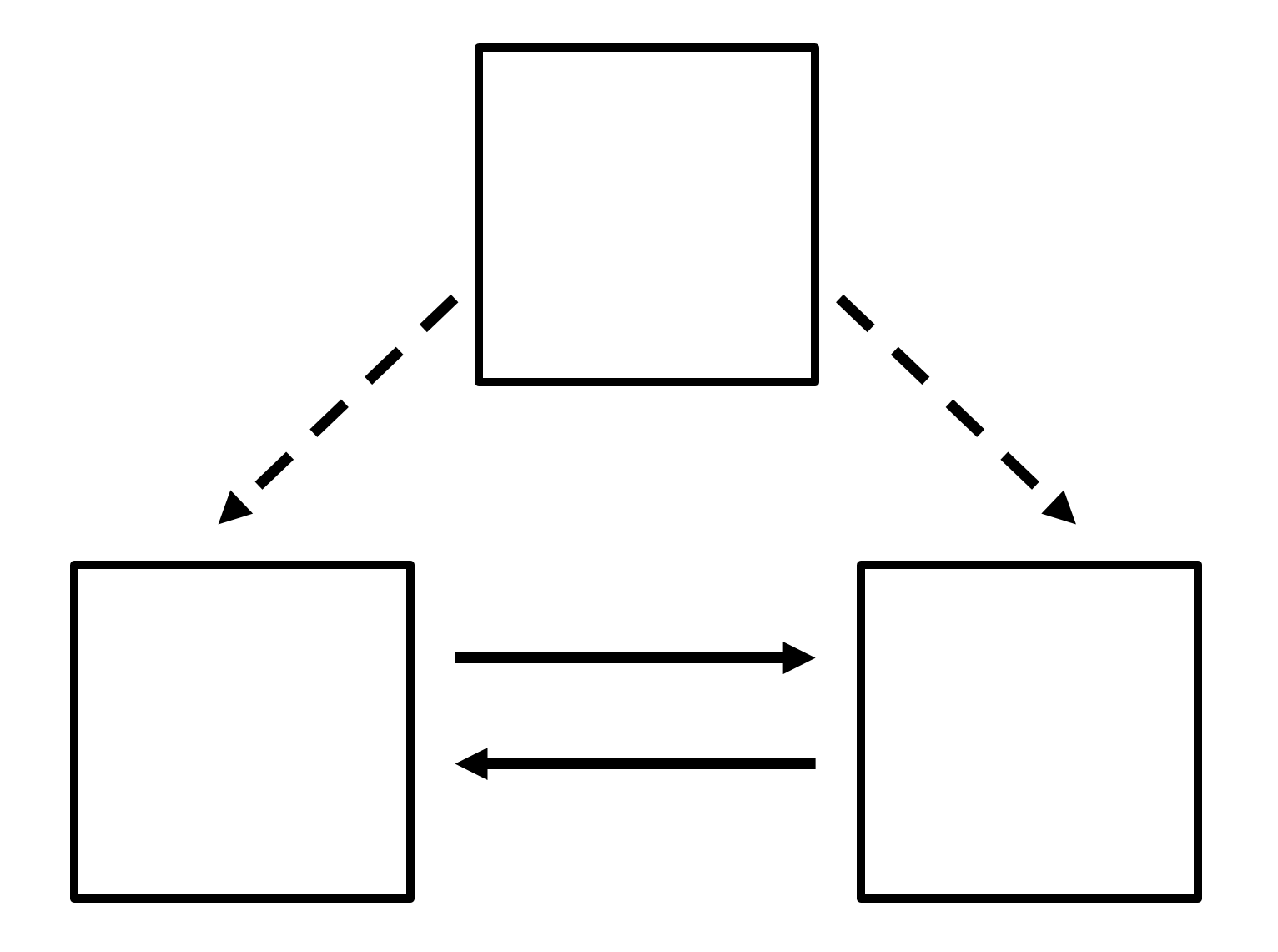}
 \put(12,20){AP$\,1$}
 \put(12,10){$\rvec{H}_1[t]$}
 \put(73,20){AP$\,2$}
 \put(73,10){$\rvec{H}_2[t]$}
 \put(44,55){CPU}
 \put(36,27){$\rvec{H}_1[t-d]$}
 \put(36,6){$\rvec{H}_2[t-d]$}
 \put(26,37){$\rvec{u}[t]$}
 \put(64,37){$\rvec{u}[t]$}
\end{overpic}
\label{fig:distributed}}
\hfil
\subfloat[]{\begin{overpic}[width=2in,tics=20]{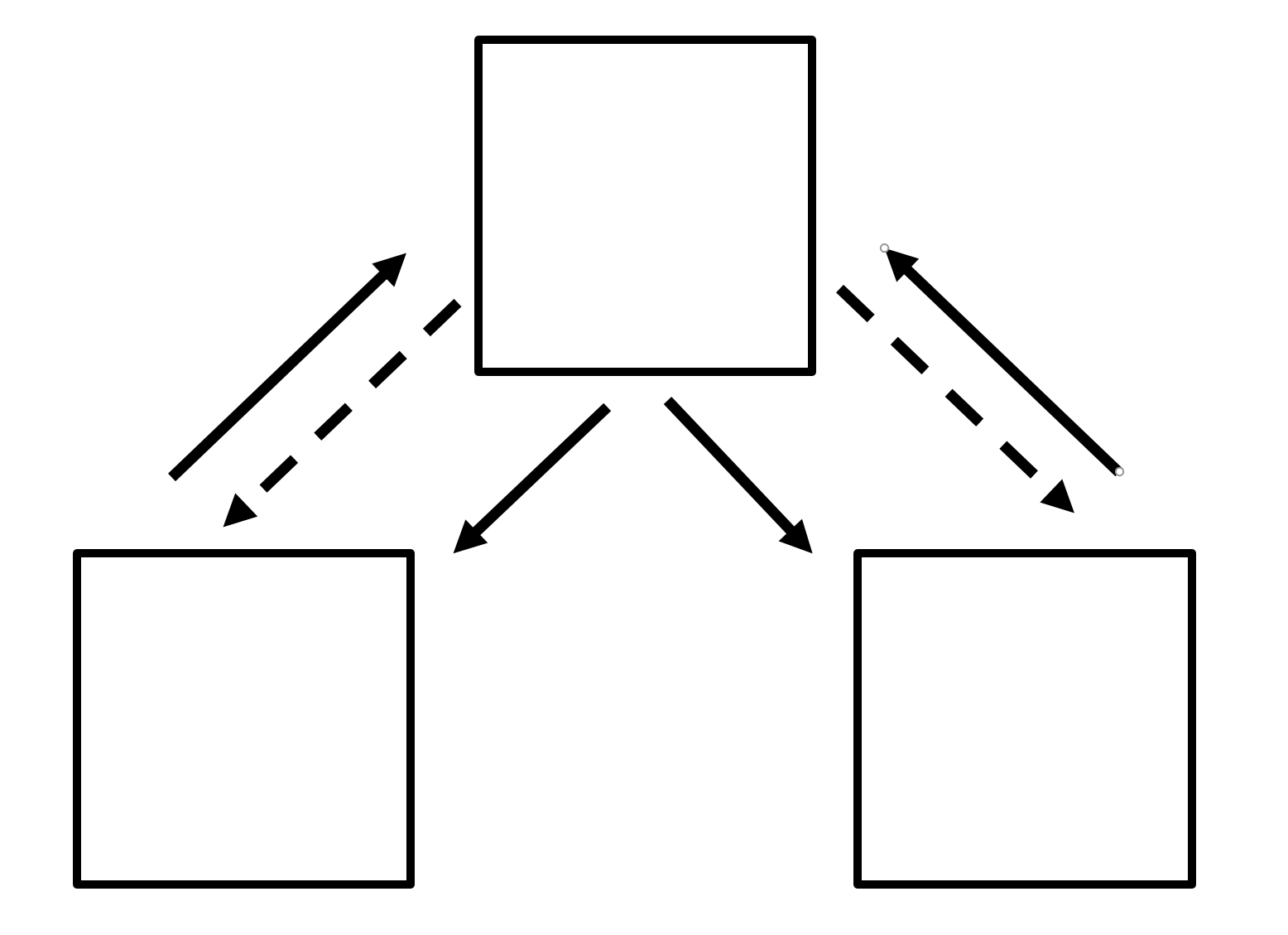}
 \put(12,20){AP$\,1$}
 \put(12,10){$\rvec{H}_1[t]$}
 \put(73,20){AP$\,2$}
 \put(73,10){$\rvec{H}_2[t]$}
 \put(44,55){CPU}
 \put(0,55){$\rvec{H}_1[t-d]$}
 \put(70,55){$\rvec{H}_2[t-d]$}
 \put(33,23){$\rvec{C}_1[t]\;\, \rvec{C}_2[t]$}
 \put(28,37){$\rvec{u}[t]$}
 \put(62,37){$\rvec{u}[t]$}
\end{overpic}
\label{fig:centralized_1}}
\hfil
\subfloat[]{\begin{overpic}[width=2in,tics=20]{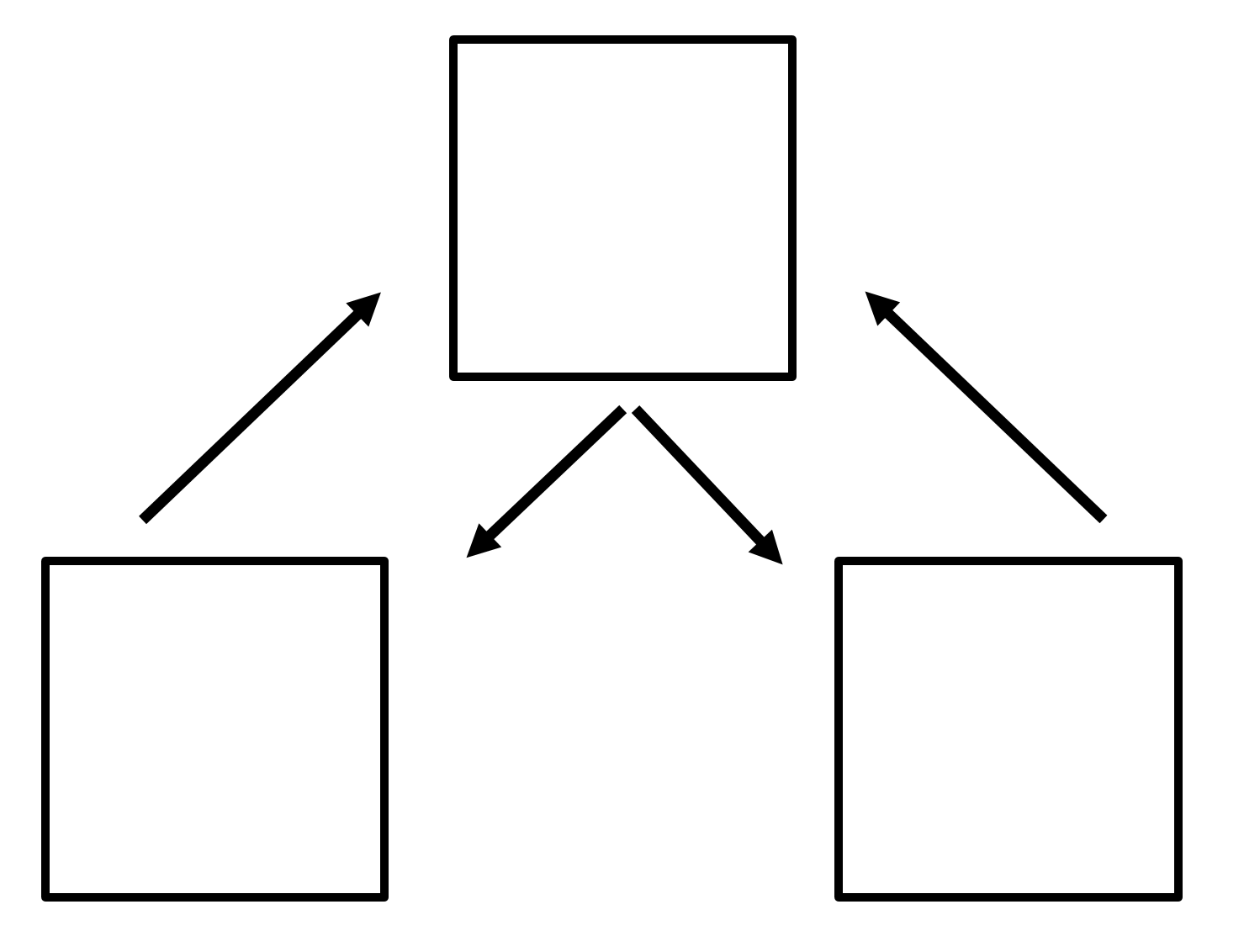}
 \put(11,20){AP$\,1$}
 \put(11,10){$\rvec{H}_1[t]$}
 \put(73,20){AP$\,2$}
 \put(73,10){$\rvec{H}_2[t]$}
 \put(44,57){CPU}
 \put(0,55){$\rvec{H}_1[t-d]$}
 \put(70,55){$\rvec{H}_2[t-d]$}
 \put(20,37){$\rvec{C}_1\rvec{u}[t]$}
 \put(60,37){$\rvec{C}_2\rvec{u}[t]$}
\end{overpic}
\label{fig:centralized_2}}}
\caption{Pictorial representation of possible implementations of the proposed team MMSE solution \eqref{eq:DTMMSE} in C-RAN architectures: (a) distributed precoding with  CSI sharing; (b) locally refined centralized precoding (compress-before-precoding); (c) locally refined centralized precoding  (compress-after-precoding). }
\label{fig:implementation}
\vspace{-0.5cm}
\end{figure*}
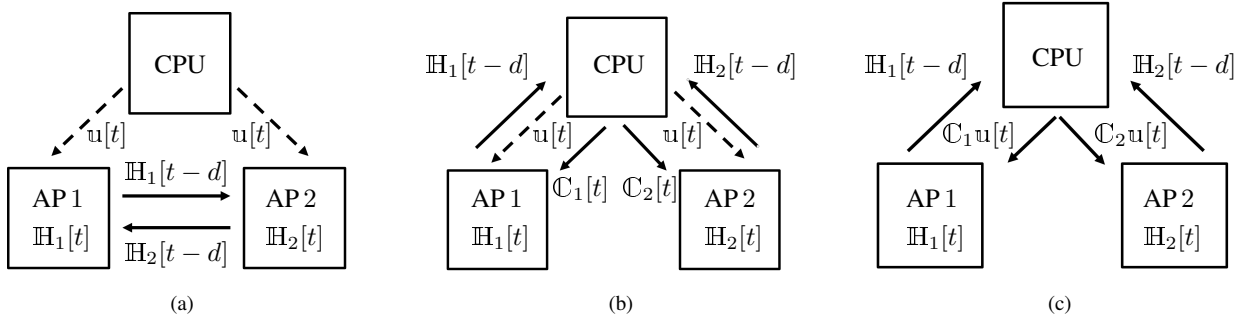
In this section we discuss the implementation of the optimal two-stage precoding structure identified in \eqref{eq:DTMMSE} by focusing on the functional split problem in C-RAN architectures \cite{kang2015fronthaul}. First, we recall that the considered solution in \eqref{eq:DTMMSE} is derived under the delayed CSI sharing model~\eqref{eq:CSI}, by assuming that the precoders are computed locally at each AP after a preliminary CSI sharing step. This model can be directly mapped to the functional split depicted in Fig.~\ref{fig:distributed}, where (from a physical layer perspective) the CPU only forms and forwards the $K$-dimensional vector of coded and modulated data streams $\rvec{u}[t]$. This is the closest implementation to the original fully distributed cell-free massive MIMO concept proposed in \cite{ngo2017cellfree}. 

However, as anticipated in Remark~\ref{rem:computation}, \eqref{eq:DTMMSE} can also be implemented by splitting the computation of the two stages between the APs and the CPU, as depicted in Fig.~\ref{fig:centralized_1} and Fig.~\ref{fig:centralized_2}. These implementations are closer to the centralized cell-free massive MIMO concept described, e.g., in \cite{demir2021foundations}. In both these functional splits, the CPU forms the precoding stages $\rvec{C}_l[t]$ based on delayed global CSI, and the APs form their local MMSE stages $\rvec{F}_l[t]$ based on timely local CSI. The difference between these two functional splits is that in Fig.~\ref{fig:centralized_1} both precoding stages are applied to the data streams by the APs, while in Fig.~\ref{fig:centralized_2} the CPU computes and forwards the $K$-dimensional intermediate signals $\rvec{C}_l[t]\rvec{u}[t]$ for each AP.\footnote{In both cases, the delay $d$ should be rather interpreted as the round-trip delay incurred by the two-way information sharing procedure.}

Choosing the best functional split is a notoriously challenging problem encompassing many different system aspects. For instance, Fig.~\ref{fig:centralized_1} and Fig.~\ref{fig:centralized_2} can be respectively interpreted as novel distributed versions of the \textit{compress-before-precoding} and \textit{compress-after-precoding} functional splits compared in \cite{kang2015fronthaul} in terms of fronthaul rate requirements. Interestingly, the results of this study may be used as a novel approach to extend the current literature on C-RAN functional splits covering fronthaul and processing delay requirements. 

We conclude this section by pointing out that, in the current form, none of the implementations in Fig.~\ref{fig:implementation} is scalable with respect to the number of UEs $K$. This is mostly due to the fact that the information to be shared and processed is proportional to $K$, as in \cite{ngo2017cellfree}. However, we remark that this issue can be significantly mitigated by omitting the sharing and processing of information that does not contribute significantly to performance, for instance because of large path loss between a certain UE-AP pair. A popular way of implementing this idea is via the user-centric network clustering paradigm \cite{demir2021foundations}. The extension of our results to this paradigm can be done as in \cite{miretti2023duality}, and it is left for the extended version of this study.     
 
\subsection{Suboptimal solutions}
\label{ssec:approximate}
\subsubsection{Local precoding}
By discarding the potentially useful information $\rvec{H}[t-d]$ in \eqref{eq:CSI}, i.e., by letting $(\forall l\in \set{L})~S_l[t] = \rvec{H}_l[t]$, the solution to Problem~\eqref{prob:TMMSE} is given by a variant of \eqref{eq:DTMMSE} with
$(\forall l\in \set{L})~\rvec{C}_l = \vec{C}_l$, where $(\vec{C}_l)_{l=1}^L$ are fixed (deterministic) precoding stages given by the solution to 
\begin{equation}\label{eq:LTMMSE}
(\forall l \in \mathcal{L})~\vec{C}_l + \sum_{j\in \set{L}\backslash \{l\}}\E[\vec{P}^{\frac{1}{2}}\rvec{H}_j^\herm\rvec{F}_j]\vec{C}_j = \vec{I}_K.
\end{equation}
This corresponds to the local \textit{team MMSE} solution derived in \cite{miretti2021team,miretti2023duality}, which we recall is an enhanced version of the known local MMSE scheme and its variants \cite{demir2021foundations}. 

\subsubsection{Centralized (delay-tolerant) precoding}
On the other extreme, by discarding  $(\rvec{H}_l[t])_{l=1}^L$ in \eqref{eq:CSI}, i.e., by letting $(\forall l\in \set{L})~S_l[t] = \rvec{H}[t-d]$, the solution to Problem~\eqref{prob:TMMSE} becomes
\begin{equation}\label{eq:CMMSE}
\rvec{T}[t] = \left(\hat{\rmat{H}}[t]\vec{P}\hat{\rmat{H}}[t]^\herm + \vec{\Psi}+\vec{\Sigma}\right)^{-1}\hat{\rmat{H}}[t]\vec{P}^{\frac{1}{2}},
\end{equation}
where $\hat{\rvec{H}}[t] \eqdef \E[\rvec{H}[t]|\rvec{H}[t-d]]$, $\vec{\Psi}\eqdef \E[(\hat{\rvec{H}}[t]-\rvec{H}[t])\vec{P}(\hat{\rvec{H}}[t]-\rvec{H}[t])^\herm]$, and $\vec{\Sigma}\eqdef \text{blkdiag}(\sigma_1\vec{I}_N,\ldots,\sigma_L\vec{I}_N)$. This essentially corresponds to the known centralized MMSE scheme in \cite{demir2021foundations} or to the delay-tolerant zero-forcing scheme based on channel prediction in \cite{schotten2022delay}, carefully optimized and adapted to our setup to ensure a fair comparison against \eqref{eq:DTMMSE}.  


\subsubsection{Na\"ive distributed precoding}
A simple baseline distributed precoding scheme that takes into account the full information in \eqref{eq:CSI} is given by letting each AP $l\in \set{L}$ locally compute a version of the centralized  solution \eqref{eq:CMMSE} based on its local estimate of the global channel state, obtained by replacing the submatrix $\hat{\rvec{H}}_l[t]\eqdef \E[\rvec{H}_l[t]|\rvec{H}_l[t-d]]$ of $\hat{\rvec{H}}[t]$ with $\rvec{H}_l[t]$. This baseline approach was termed \textit{na\"ive} precoding in the early works on distributed precoding \cite{bazco2022decentralized}, since it essentially corresponds to letting each AP believe that its  information on $\rvec{H}[t]$ is the same information at all APs.    

\subsubsection{Structure-aware distributed precoding} As an alternative to the above baseline distributed precoding scheme that takes into account the structure of the optimal solution, we propose a variant of \eqref{eq:DTMMSE} based on approximating the coefficients of the linear system in \eqref{eq:lin_syst} as $\E[\vec{P}^{\frac{1}{2}}\rvec{H}_l[t]^\herm\rvec{F}_l[t]|\rvec{H}_l[t-d]] \approx$
\begin{equation*}
\vec{P}^{\frac{1}{2}}\hat{\rvec{H}}_l[t]^\herm\left(\hat{\rvec{H}}_l[t]\vec{P}\hat{\rvec{H}}_{l}[t]^\herm + \vec{\Psi}_l + \sigma_l\vec{I}_N\right)^{-1}\hat{\rvec{H}}_{l}[t]\vec{P}^{\frac{1}{2}},
\end{equation*}
where $\vec{\Psi}_l \eqdef \E[(\hat{\rvec{H}}_l[t]-\rvec{H}_l[t])\vec{P}(\hat{\rvec{H}}_l[t]-\rvec{H}_l[t])^\herm]$.
This is similar to the centralized approach described above, but confined to the computation of the precoding stages $(\rvec{C}_l[t])_{l=1}^L$. 


\section{Numerical simulations and conclusions}
\label{sec:sim}
We consider a network composed by $K=50$ UEs are uniformly distributed within a squared service area of size $0.5\times 0.5~\text{km}^2$, and $L=16$ regularly spaced APs with $N=4$ antennas each. By neglecting for simplicity spatial correlation, we let $\rvec{h}_{l,k}$ be independently distributed as $\rvec{h}_{l,k} \sim \CN\left(\vec{0}, \gamma_{l,k}\vec{I}_N\right)$, where $\gamma_{l,k}>0$ denotes the channel gain between AP $l$ and UE $k$. We follow the same 3GPP-like path-loss model adopted in \cite{demir2021foundations} for a $2$ GHz carrier frequency:
\begin{equation*}
\gamma_{l,k} = -36.7 \log_{10}\left(D_{l,k}/1 \; \text{m}\right) -30.5 + Z_{l,k} -\sigma^2 \quad \text{[dB]},
\end{equation*}
where $D_{l,k}$ is the distance between AP $l$ and UE $k$ including a difference in height of $10$ m, and $Z_{l,k}\sim \set{N}(0,\rho^2)$ [dB] are shadow fading terms with deviation $\rho = 4$. The shadow fading is correlated as $\E[Z_{l,k}Z_{j,i}]=\rho^22^{-\frac{\delta_{k,i}}{9 \text{ [m]}}}$ for all $l=j$ and zero otherwise, where $\delta_{k,i}$ is the distance between UE $k$ and UE~$i$. The noise power is $\sigma^2 = -174 + 10 \log_{10}(B/1\;\text{Hz}) + F$ [dBm], where $B = 20$ MHz is the bandwidth, and $F = 7$ dB is the noise figure. 

The time evolution of the channel is modeled as in \cite{heath2013aging} and many related studies by assuming that each $\{\rvec{h}_{l,k}[t]\}_{t\in\stdset{Z}}$ is a zero-mean stationary ergodic complex Gaussian process, where the joint distribution of $(\rvec{h}_{l,k}[t],\rvec{h}_{l,k}[t-d])$ is fully characterized by the autocovariance matrix $\E[\rvec{h}_{l,k}[t]\rvec{h}_{l,k}[t-d]^\herm]=r_{l,k}\gamma_{l,k}\vec{I}_N$ for a given autocorrelation coefficient $r_{l,k}\in [0,1]$. The autocorrelation coefficient can be used to model the joint impact of the CSI sharing delay $d$ and UE mobility, for example by using Clarke's model $r_{l,k} = J_0(2\pi\nu_{l,k}Td)$ as in \cite{heath2013aging}, where $\nu_{l,k}$ denotes the Doppler spread for the $(l,k)$th AP-UE pair, and $T$ is the symbol time. We focus on two representative scenarios obtained by letting $(\forall l \in \set{L})(\forall k \in \set{K})~r_{l,k}= r \in \{0.99,0.9\}$, which, according to Clarke's model, can be mapped to pedestrian mobility $(\nu_{l,k} \approx 10$~Hz) for all UEs and delay $Td \in \{1~\text{ms},10~\text{ms}\}$. 

\subsection{Figure-of-merit}
We evaluate performance in terms of downlink ergodic achievable rates estimated by the hardening inner bound \cite{demir2021foundations} for a given (network-wide) sum power $P = 5$ W. To facilitate the connection with the solution to Problem~\eqref{prob:TMMSE}, we leverage the known uplink-downlink duality principle for fading channels (see, e.g., \cite{demir2021foundations}) and compute the downlink rate of each UE~$k\in \set{K}$ for a given precoding scheme $\rvec{T}$ and for some downlink power allocation policy using the equivalent expressions
\begin{equation*}
R_k(\rvec{T},\vec{p}) \eqdef \log_2(1+\mathrm{SINR}_k(\rvec{t}_k,\vec{p})),
\end{equation*}
\begin{equation*}
\mathrm{SINR}_k(\rvec{t}_k,\vec{p}) \eqdef \resizebox{0.69\linewidth}{!}{$\dfrac{p_k|\E[\rvec{h}_k^\herm\rvec{t}_k]|^2}{p_k\V(\rvec{h}_k^\herm\rvec{t}_k)+\underset{j\neq k}{\sum} p_j\E[|\rvec{h}_j^\herm\rvec{t}_k|^2]+\E[\|\rvec{t}_k\|^2]}$},
\end{equation*}
corresponding to the achievable rates over a virtual uplink channel for some virtual uplink powers $\vec{p} \in \stdset{R}_+^K$ satisfying $\sum_{k=1}^Kp_k = P$.  We recall that the uplink-downlink duality principle guarantees the existance of a downlink power allocation such that the same rates are also achievable in the downlink using $\rvec{T}$. It can be shown that the above uplink rates are maximized by the solution to Problem~\eqref{prob:TMMSE} with problem parameters $\vec{p}$ equal to the uplink powers and $\vec{\sigma}=\vec{1}$ \cite{miretti2021team,miretti2022joint}. As an example, we choose a fractional power allocation policy with exponent $-1$ \cite{lozano2020fractional}, i.e., we let  $(\forall k\in \set{K})~p_k \propto (\sum_{l=1}^L\gamma_{l,k})^{-1}$, which approximates a max-min fair policy. 

\subsection{Results}
\begin{figure*}[!t]
\centering
\subfloat[]{\includegraphics[width=0.84\columnwidth]{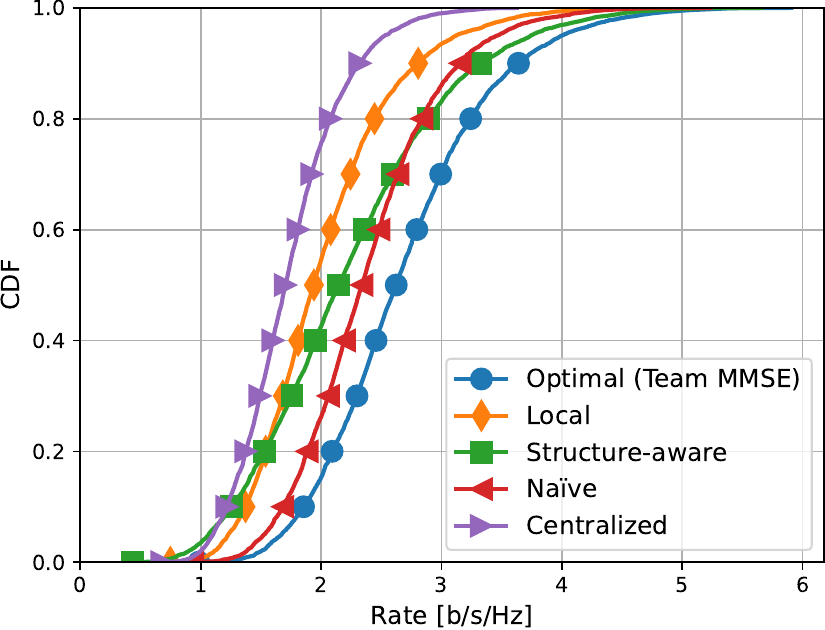}
\label{fig_first_case}}
\hfil
\subfloat[]{\includegraphics[width=0.84\columnwidth]{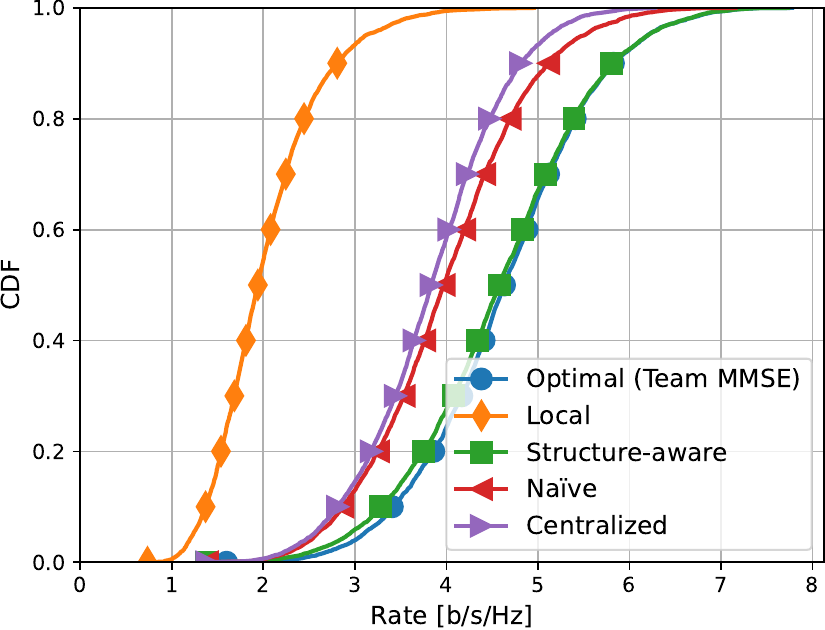}
\label{fig_second_case}}
\caption{Cumulative density function of  downlink ergodic rates achieved by the optimal team MMSE solution \eqref{eq:DTMMSE} and the suboptimal designs in Sect.~\ref{ssec:approximate}, assuming pedestrian mobility and a CSI sharing delay of (a) $10$ ms ($r = 0.9$) and (b) $1$ ms ($r = 0.99$).}
\label{fig_sim}
\vspace{-0.57cm}
\end{figure*}
Figure~\ref{fig_sim} shows the cumulative density function of the ergodic rates achieved by different precoding schemes, approximated using $100$ independent user drops and $100$ independent realizations of $(\rvec{H}[t],\rvec{H}[t-d])$ per user drop. In particular, we compare the performance of the optimal Team MMSE solution \eqref{eq:TMMSE} against the suboptimal solutions described in Section~\ref{ssec:approximate}. We remark that the schemes termed \textit{centralized} and \textit{local} correspond, respectively, to the best known schemes for the centralized and (fully) distributed cell-free massive MIMO implementations described in \cite{demir2021foundations}. The Team MMSE solution is computed via the numerical method described in  Section~\ref{ssec:computation}. 

A first important observation is that, even for pedestrian mobility, a CSI sharing delay of $10$ ms ($r = 0.9$, Figure~\ref{fig_first_case}) can significantly degrade the performance of centralized precoding to the point where it becomes noticeably worse than the performance of local precoding. Moreover, a second important observation is that that the proposed Team MMSE precoding scheme largely outperforms both centralized and local precoding in all the considered scenarios. Perhaps surprisingly, the gains are significant even for a relatively small CSI sharing delay of $1$ ms ($r = 0.99$, Figure~\ref{fig_second_case}), a scenario where centralized precoding is still quite effective and outperforms local precoding. This suggests that the APs' local interference management capabilities based on timely local CSI should not be neglected in most practical scenarios. Indeed, our results suggests that significant rate gains can be achieved by a careful distributed precoding design, such as the proposed Team MMSE scheme, that is able to merge the benefits of local and centralized interference management. 

Finally, Figure~\ref{fig_second_case} shows that, for small CSI sharing delays, the aforementioned rate gains can be achieved by simple approximations of the Team MMSE scheme. In particular, in the considered scenario, we observe that the proposed structure-aware distributed precoding scheme is able to exploit the benefits of local and centralized interference management, at a significantly lower computational cost than the Team MMSE scheme. However, Figure~\ref{fig_first_case} shows that this may not be the case for higher CSI sharing delays, since the Team MMSE solution shows non-negligible performance gains.

\bibliographystyle{IEEEbib}
\bibliography{IEEEabrv,refs}

\end{document}